\documentclass[letterpaper, 10 pt, conference]{ieeeconf}  

\IEEEoverridecommandlockouts                              
\usepackage{graphics} 
\usepackage{xcolor}
\usepackage{amsmath} 
\usepackage{amssymb}  
\usepackage{mathtools}
\usepackage{changes}
\usepackage{balance}
\usepackage{paralist}

\usepackage{tikz}
\usetikzlibrary{arrows, arrows.meta}
\usetikzlibrary{automata,positioning}
\usetikzlibrary{shapes.geometric}

\title{\LARGE \bf
Integrated Control and Active Perception in POMDPs for Temporal Logic Tasks and Information Acquisition
}

\usepackage{tikz}  \usepackage{xcolor}
\definecolor{light-gray}{gray}{0.95}
\definecolor{orange}{rgb}{1,0.5,0}
\definecolor{darkgreen}{rgb}{0,204,0}

\usetikzlibrary{shapes,shapes.geometric,arrows,fit,calc,positioning,automata,}

\author{Chongyang Shi$^1$, Michael R. Dorothy$^2$, and Jie Fu$^1$
\thanks{1. Chongyang Shi, Jie Fu with the Department of Electrical and Computer Engineering, University of Florida, Gainesville, FL, USA. {\tt\small \{c.shi, fujie\}@ufl.edu }.}%
\thanks{2. Michael R. Dorothy with DEVCOM Army Research Laboratory, Adelphi, MD, USA. {\tt\small michael.r.dorothy.civ@army.mil}.}%
}
\usepackage[printonlyused]{acronym}

\newif\ifuseboldmathops
\newif\ifuseittextabbrevs
\useboldmathopstrue   
\DeclareMathOperator*{\optmin}{\textrm{minimize}}

\ifuseittextabbrevs
	
	\newcommand{\eg}{{\it e.g.}}
	\newcommand{\ie}{{\it i.e.}}

\else
	
	\newcommand{\eg}{e.g.}
	\newcommand{\ie}{i.e.~}

\fi

\ifuseboldmathops
	\newcommand{\reals}{\mathbf{R}}

\else
	\newcommand{\reals}{\mathbb{R}}

\fi

\ifuseboldmathops

\else

\fi

\ifuseboldmathops

\else

\fi

\newcommand{\obs}{E}

\newcommand{\dist}[1]{\mathcal{D}(#1)}

\newtheorem{definition}{Definition}
\newtheorem{proposition}{Proposition}

\newtheorem{problem}{Problem}

\newtheorem{remark}{Remark}

\newtheorem{lemma}{Lemma}

\newcommand{\calA}{\mathcal{A}}
\newcommand{\calAP}{\mathcal{AP}}
\newcommand{\calS}{\mathcal{S}}
\newcommand{\calM}{\mathcal{M}}

\newcommand{\calO}{\mathcal{O}}

\newcommand{\calF}{\mathcal{F}}
\newcommand{\expect}{\mathbb{E}}

\newcommand{\nat}{\mathbf{N}}

\newcommand{\dfa}{\mathbb{A}}

\newcommand{\truev}{\mathsf{true}}

\newcommand{\Always}{\Box \, }
\newcommand{\Eventually}{\lozenge \, }
\newcommand{\Next}{\bigcirc \, }
\newcommand{\Until}{\mbox{$\, {\sf U}\,$}}

\newcommand{\init}{{\iota}}

\acrodef{mdp}[MDP]{Markov decision process}
\acrodef{pomdp}[POMDP]{partially observable Markov decision process}
\acrodef{lmdp}[LMDP]{Labeled Markov decision process}
\acrodef{lpomdp}[labeled-POMDP]{Labeled partially observable Markov decision process}

\acrodef{asw}[ASW]{Almost-Sure Winning}
\acrodef{ltlf}[LTL$_f$]{Linear Temporal Logic over Finite Traces}
\acrodef{ltl}[LTL]{linear temporal logic}
\acrodef{scltl}[co-safe LTL]{syntactically co-safe Linear Temporal Logic}
\acrodef{dfa}[DFA]{deterministic finite automaton}
\acrodef{des}[DES]{discrete event system}
\acrodef{hmm}[HMM]{hidden Markov model}

\newcommand{\matO}{{\mathbf{O}}}
\newcommand{\matT}{{\mathbf{T}}}
\newcommand{\matA}{{\mathbf{A}}}

\begin{document}

\maketitle
\thispagestyle{empty}
\pagestyle{empty}

\begin{abstract}This paper studies the synthesis of a joint control and active perception policy for a stochastic system modeled as a partially observable Markov decision process (POMDP), subject to temporal logic specifications. The POMDP actions influence both system dynamics (control) and the emission function (perception). Beyond task completion, the planner seeks to maximize information gain about certain temporal events (the secret) through coordinated perception and control.
To enable active information acquisition, we introduce minimizing the Shannon conditional entropy of the secret as a planning objective, alongside maximizing the probability of satisfying the temporal logic formula within a finite horizon. Using a variant of observable operators in hidden Markov models (HMMs) and POMDPs, we establish key properties of the conditional entropy gradient with respect to policy parameters. These properties facilitate efficient policy gradient computation. We validate our approach through graph-based examples, inspired by common security applications with UAV surveillance.
\end{abstract}

\section{INTRODUCTION}
In many mission-critical systems, control objectives often integrate both reward-based and information-theoretic performance metrics. 
Consider an autonomous robot tasked with patrolling a large building to ensure security. The robot must follow a predefined sequence of waypoints, and upon detecting an intrusion, it must perform a reactive mission, such as intercepting the intruder.
Such a task can be described in temporal logic \cite{mannaTemporalLogicReactive1992a}. While it is natural to formulate a planning objective to maximize the probability of satisfying the temporal logic. One may be interested in knowing whether an intrusion occurred in the past. This planning problem is challenging given the imperfect information present in the environment (e.g., limited visibility, sensor noise, and other environmental factors).   Lacking access to the underlying state trajectory, the agent can be uncertain about the sequence of temporal events that have occurred, making it uncertain regarding if a temporal logic formula is satisfied or certain critical temporal events have occurred.

Motivated by these applications, this paper investigates a class of \ac{pomdp}s with \emph{both controllable transition and controllable observation functions}. The planning objective is twofold: first, the agent aims to maximize the probability of satisfying a task expressed in linear temporal logic over finite traces (LTLf); second, the agent seeks to maximize the information gained during a finite horizon regarding the occurrence of critical temporal events. The model explicitly allows actions to actively elicit information in stochastic environments to minimize uncertainty in critical events assessment. Specifically, we introduce an information-theoretic performance metric, measured by the remaining uncertainty regarding the critical event estimation given observations. The objective is to design a \emph{joint control and active perception policy in POMDPs} that maximizes the probability of task completion while minimizing the uncertainty in inferring the occurrence of a critical event. This critical event can be a subformula of the original temporal logic formula, for example, whether the intrusion occurred in the past, regardless if the robot interdicted the intruder or not.
 
\paragraph*{Related work} 
 Planning and stochastic control  with temporal logic objectives has been studied extensively, including 
belief-based planning \cite{leahyControlBeliefSpace2019}, point-based methods \cite{bouton2020point}, and finite-state-controller synthesis  
\cite{sharanFiniteStateControl2014,kalagarlaOptimalControlPartially2022, carrCounterexampleGuidedStrategyImprovement2019}. A typical approach is to construct a product MDP between the original \ac{mdp} and the automata representation for the temporal logic formula, and then extend existing POMDPs solutions to solve for the temporal logic constrained planning in the product \ac{pomdp}. Aside from traditional POMDP methods, in \cite{andriushchenkoInductiveSynthesisFinitestate2022}, the authors presented
a counter-example-guided inductive synthesis to directly explore a design space that compactly represents potential finite-state controllers for POMDPs subject to temporal logic constraints. Recent methods also employ deep reinforcement learning for POMDP with temporal logic constraints. In \cite{carrCounterexampleGuidedStrategyImprovement2019a}, the authors
introduces a solution for POMDP planning by first trainning a recurrent neural network (RNN) policy for the POMDP  and then restricting the RNN policy to finite-state controllers, which is then verified using formal verification techniques to provide provable guarantees against temporal logic specifications.  
To our knowledge, existing POMDP planning has not considered information-theoretic objectives related to the inference of temporal properties given partial observations.  

 \ac{pomdp} models are widely used for robotic motion planning with imperfect sensing modality and active information acquisition tasks.  
In \cite{Bai2014}, the authors introduce joint planning and active perception using a continuous-state POMDP whose observation function is influenced by the action as well, similar to our modeling. However, they  consider reward-based performance measures and show active perception plays a key role in improving total reward.    In \cite{Ghasemi2019}, the author add a limited budget for the ability of information gathering of the agent and use the point-based value iteration method to obtain a generalized greedy policy that is proved to be near-optimal. 
Recent literature \cite{pulido2025uncertaintyawareguidancetargettracking} uses the maximum expected entropy reduction (EER) method to track a target whose motion is modeled by a transformer-based neural network. The EER method estimate the target state by the Shannon entropy and optimizes the entropy reduction to reduce the uncertainty of the target.
In our recent work \cite{shi2024active}, we study active perception to infer the initial state of a \ac{hmm} and also  measure the uncertainty of the initial state  by the Shannon entropy. However, the perception action cannot influence the transition function but only the emission function of the \ac{hmm}.

In formal synthesis, 
active perception has been introduced for qualitative planning for partially-observable non-deterministic systems  \cite{fuSynthesisJointControl2016,liTemporalLogicTask2024}. These planning methods involve offline computation of \emph{winning} belief states that can ensure the satisfaction of a given temporal logic objective and online-active-sensing decisions to reach these winning belief states. Our work differs from existing methods in temporal logic planning with active perception in two aspects: First, our objective is quantitative instead of qualitative (\ie, satisfy the objective with probability one or surely); second, our perception objective needs not to be coupled with the control objective. Aside from the task's success, the agent aims to collect informative observations for other inference tasks.

\paragraph*{Our contributions}

We develop a policy gradient-based approach to compute a (local) optimal joint control and active perception policy. Building on our prior work on active perception in \ac{hmm}s \cite{shi2024active},  we develop a method for computing the joint control and perception policy gradient of the conditional entropy of the secret state given observations. The approach leverages the input-output observable operator models \cite{thonLinksMultiplicityAutomata2015} for efficient computation of posterior estimates in critical temporal events. 

In section~\ref{sec:pre}, we provide the definition of \ac{lpomdp} and the \ac{ltlf}. In section~\ref{sec:problem}, we formulate the joint control and perception problem as an optimization problem. The objective function is defined as a weighted sum of the entropy of the secret and the probability of completing a task. In section~\ref{sec:compute}, we develop the policy gradient algorithm to solve the optimization problem. In section~\ref{sec:simulation}, we demonstrate the effectiveness of the proposed policy gradient algorithm by a graph-like \ac{pomdp} with task specifications inspired by robotic surveillance and patrolling tasks.

\section{PRELIMINARIES}
\label{sec:pre}

\noindent \textbf{Notation} The set of real numbers is denoted by $\reals$. Random variables will be denoted by capital letters, and their realizations by lowercase letters (\eg, $X$ and $x$).  A sequence of random variables and their realizations with length $T$ are denoted as $X_{0:T}$ and $x_{0:T}$. The notation $x_i$ refers to the $i$-th component of a vector $x \in \reals^{n}$ or to the $i$-th element of a sequence $x_0, x_1, \ldots$, which will be clarified by the context. 
Given a finite set $\mathcal{S}$, let $\dist{\mathcal{S}}$ be the set of all probability distributions over $\mathcal{S}$. The set $\mathcal{S}^{T}$ denotes the set of sequences with length $T$ composed of elements from $\mathcal{S}$, and $\mathcal{S}^\ast$ denotes the set of all finite sequences generated from $\mathcal{S}$. The empty string in $\mathcal{S}^\ast$ is denoted by $\varnothing$.

We consider the planning problem modeled as a \ac{lpomdp}. 
\begin{definition}[\ac{lpomdp}]
A \ac{lpomdp} is defined as a tuple $
M = \langle \calS , \calO, \calA, P, E, S_0,  \calAP, L  \rangle$
where 
\begin{inparaenum}
    \item $\calS$ is a finite  state space. 
    \item $\calO$ is a finite set of observations.
    \item $\calA $ is a finite set of actions including \emph{control actions} and \emph{perception actions}. 
    \item $P:\calS \times \calA \rightarrow \dist{\calS}$ is the probabilistic transition function.
    \item $E: \calS \times \calA \rightarrow \dist{\calO}$ is the emission function (observation function) that takes a state $s$, a perception action $a$, and outputs a distribution over observations.  
    \item $S_0$ is a random variable representing the initial state. The distribution of $S_0$ is denoted by $\mu_0$. And $\mathcal{S}_0$ denotes the set of initial states.
    \item $\calAP$ is the set of atomic propositions; and
    \item $L: S\rightarrow 2^{\calAP}$ is the labeling function that maps a state to a set of atomic propositions that evaluate true at that state. 
\end{inparaenum}
\end{definition} 

A finite play $\rho = s_0 a_0 s_1 a_1 s_2\ldots s_n $ is a sequence of interleaving states and players' actions such that $P(s_{i+1}|s_i,a_i) > 0$ for all integers $0 \leq i \leq n-1$. 
The labeling of the play $\rho$, denoted $L(\rho)$, is defined as $L(\rho) = L(s_0)L(s_1)\ldots L(s_n)$. That is, the labeling function omits the actions from the play and applies to states only.

%
We consider the planning objective in the \ac{lpomdp} expressed by \ac{ltlf} formulas.
\begin{definition}[\ac{ltlf} Syntax \cite{de2013linear}]

Given a set of atomic propositions $\calAP$, a \ac{ltlf} formula over $\calAP$ is defined using the following grammar:
\[ \varphi := \top \mid \bot \mid p \mid \neg p \mid  \varphi \land \varphi \mid \varphi \vee \varphi \mid \Next \varphi \mid \varphi_1 \Until \varphi_2 , \] where $\top,\bot$ are universally true and false, respectively, $p \in \calAP$ is an atomic proposition, $\neg$, $\land$ and $\vee$ are the Boolean operators \textit{negation}, \textit{disjunction}, and \textit{conjunction}, respectively, and $\Next$ and $ \Until$   denote the temporal operators \textit{next}, and \textit{until}  respectively.
\end{definition}
Additionally, we also introduce two temporal operators:  ``Eventually'' ($\Eventually$) is defined as $\Eventually \varphi := \truev \Until \varphi$; and   ``Globally'' ($\Always$) is defined as $\Always \varphi := \neg \Eventually \neg \varphi$.  For formal semantics of \ac{ltlf}, see~\cite{de2013linear}.

The language of \ac{ltlf} formula $\varphi$ can be represented by the set of words accepted by an automaton defined as follows:

\begin{definition}[\ac{dfa}]
\label{def:dfa}
A \ac{dfa} is a tuple $\dfa = (\mathcal{Q}, \Sigma, \delta, \init, F )$ in which (1) $\mathcal{Q}$ is the set of states; (2) $\Sigma$ is the alphabet (set of input symbols); (3) $\delta: \mathcal{Q} \times \Sigma \rightarrow \mathcal{Q}$ is a deterministic transition function and is complete \footnote{For any $\mathcal{Q} \times \Sigma$, $\delta(q,\sigma)$ is defined. An incomplete transition function can be completed by adding a sink state and redirecting all undefined transitions to that sink state.}; (4) $\init$ is the initial state; and (5) $F \subseteq \mathcal{Q}$ is the set of accepting states. 
\end{definition}

The transition function $\delta$ is extended as $\delta(q, \sigma \cdot w) = \delta(\delta(q,\sigma), w)$ where the state $q \in \mathcal{Q}$ and input $\sigma \in \Sigma$. %
A word $w = w_0 w_1 \ldots w_n \in \Sigma^\ast$ is accepted by $\dfa$ if and only if $\delta(\init, w)\in F$. 
The set of words accepted by $\dfa$ is called the language of $\dfa$, denoted by $\mathcal{L}(\dfa)$.
Formally, $\mathcal{L}(\dfa) = \{ w \in \Sigma^* \mid \delta(q_0, w) \in F \}$.
For notation simplicity, let $\Sigma \coloneqq 2^\calAP$. 

In addition to satisfying a task in temporal logic,  the agent seeks to estimate key properties relevant to the task. For instance, based on its partial observations, the agent may aim to determine whether the labeling of the trajectory complies with a given temporal logic property. We present the following informal problem definition:
\begin{problem}
Given a \ac{lpomdp} $M$ and a \ac{ltlf} formula expressed by a \ac{dfa} $\dfa = (\mathcal{Q}, \Sigma, \delta, \init, F , F_{sec})$ where $F\subseteq Q$ be a set of \emph{final} states and $F_{sec} \subseteq Q$ be a set of \emph{secret} states, a finite horizon $T>0$,  the joint control and active perception planning objective is to \begin{inparaenum}
    \item maximize the probability in satisfying the \ac{ltlf} formula within the finite time;
    \item minimizing the uncertainty regarding whether the labeling sequence $\sigma_{0:T}$  of the trajectory leads to a secret state, i.e., whether \( \delta(\init, \sigma_{0:T}) \in F_{sec} \).
\end{inparaenum}
\end{problem}




\section{PROBLEM FORMULATION}
\label{sec:problem}

In this section, we develop a mathematical formulation of the joint control and active perception problem.

\begin{definition}[Product POMDP]
\label{def:product_mdp}
Given the labeled-POMDP $M=\langle \calS , \calO, \calA, P, E, S_0,  \calAP, L  \rangle$, and a \ac{dfa} $\mathbb{A} = (\mathcal{Q},2^\calAP, \delta, \init, F ,F_{sec})$ describing P1's  task objective $\varphi$ and $F_{sec}$ is a set of secret states, the product POMDP a tuple 
\[
\mathcal{M} = (\mathcal{V}, \calO, \mathcal{A}, \Delta, E, v_0, \calF, \calF_{sec} )
\]
in which
\begin{itemize}
    \item $\mathcal{V} = \mathcal{S} \times \mathcal{Q}$ is the state space, where each state $(s,q)$ includes a \ac{pomdp} state $s \in \mathcal{S}$, and an automaton state $q \in \mathcal{Q}$; 
    \item $\calO$ is a finite set of observations;
    \item $\mathcal{A}$ is the action space, same as in $M$;
    \item $\Delta: \mathcal{V} \times \mathcal{A} \rightarrow \dist{\mathcal{V}}$ is the probabilistic transition function such  given a state $v \coloneqq (s,q) \in \mathcal{V}$, for any state $s' \in \mathcal{S}$ and an action $a \in \mathcal{A}$, 
    \[
    \Delta((s',q')| (s, q), a) =  P(s'|s,a)
    \]
    where $q'=\delta(q, L(s'))$. Else, $\Delta((s',q')| (s, q), a) = 0$.
    \item $E: \mathcal{V} \times \mathcal{A} \rightarrow \dist{\calO}$ is the emission function (observation function) that takes an augmented state $v$, a perception action $a$, and outputs a distribution over observations. $E(o|v) = E(o|s)$ for $v = (s, q) \in \mathcal{V}$. 
    \item $\mu_0$ is the initial state distribution that the initial state $v_0 \coloneqq (s_0,q_0) \sim \mu_0$ with $s_0 \sim \chi_0, q_0 = \delta(\init, L(s_0))$. 
  \item   $\calF \coloneqq \calS\times F$ is a set of final product  states.
  \item $\calF_{sec} \coloneqq \calS\times F_{sec}$ is a set of secret product states.
\end{itemize}
\end{definition} 

A \emph{non-stationary, observation-based} joint control and perception policy is a function $\pi: \calO^\ast \to \dist{\calA}$ that maps a history $o_{0:t}$ of observations to a distribution $\pi(\cdot|o_{0:t})$ over actions.

For a given product \ac{pomdp} $\mathcal{M}$, a policy $\pi$ induces a discrete stochastic process $\mathcal{M}_\pi \coloneqq \{V_t, A_t,  O_t, t\in \nat\}$, where $V_t \in \mathcal{V}$ and $O_t \in \mathcal{O}$ are the underlying hidden state and observation at the $t$-th time step, and   
\begin{equation}
\begin{aligned}
    &P^\pi(O_t = o | O_{0:t-1}= o_{0:t-1}, V_t=v_t) =\\  
    &\sum_{a\in \calA} E(o| v_t, a )\pi(a| o_{0:t-1}), \forall t>0.
\end{aligned}
\end{equation}
when the policy $\pi$ is understood from the context, we write $P$ instead of $P^\pi$ for clarity.


 For the stochastic process $\calM_\pi$,
  for any finite horizon $T$, the agent's partial observation about a path includes a sequence $o_{0:T}$ of observations and a sequence $a_{0:T}$ of actions. We denote the agent's information by $y_{0:T}=(o_{0:T}, a_{0:T})$.   We refer to $y_{0:T}$ as an observation sequence with the understanding that it is the joint observation and action sequence.
  We introduce two sets of random variables, which will help to facilitate the formulation of the optimization problem. 
\begin{definition} 
Given the product \ac{pomdp} $\calM$ and a policy $\pi$ induce a discrete stochastic process $\calM_\pi = \{V_t, O_t, t\in \nat\}$ with $T>0$ a finite number. Let $ \calF   $ be a set of \emph{final} product states and $\calF_{sec} $ be a set of \emph{secret} product states. For any $0\le t\le T$, let the random variable $W_t$ and $Z_t$ be defined by 
\[ W_t = \mathbf{1}_{\calF} (V_t).\] 
and
\[
Z_t =  \mathbf{1}_{\calF_{sec}} (V_t).
\]
That is, $W_t$  (resp. $Z_t$) is the random variable representing if the $t$-th state is in the set $\calF $ of final states (resp. the set $\calF_{sec}$ of secret states). The random variables $W_t$ and $Z_t$ are binary $\mathcal{W} = \{0,1\}$, $\mathcal{Z} = \{0,1\}$, for each $0 \le t \le T$.
\end{definition}

By the product construction, for any path $\rho = v_0 a_0 v_1 a_1,\ldots v_T$ in the product \ac{pomdp}, $v_T \in \calF $ means the \ac{ltlf} formula $\varphi$ is satisfied, and thus $W_T =1$.  Similarly, if $v_T \in  \calF_{sec}$, then the secret property is satisfied, and thus $Z_T = 1$. Thus, maximizing the probability of satisfying $\varphi$ is equivalent to maximizing the probability of $W_T=1$.
 
 Entropy measures are commonly employed to quantify the uncertainty about a random variable \cite{khouzani2017leakage}.  The conditional entropy of a  random variable $X_2$ given another random variable $X_1$ is defined by
\[
H(X_2|X_1) =   -\sum_{x_1\in \mathcal{X}}\sum_{x_2\in \mathcal{X}} p(x_1,x_2) \log p(x_2|x_1).
\]
The conditional entropy measures the uncertainty about $X_2$ given knowledge of $X_1$.  
A lower conditional entropy makes it easier to learn $X_2$ from observing a sample of $X_1$.
In this context, the uncertainty regarding whether a secret state in \( \calF_{sec} \) is reached at time \( T \) can be quantified using the conditional entropy of \( Z_T \), given the knowledge of the \ac{lpomdp}, the policy \( \pi \), and the observation sequence \( o_{0:T} \). Accordingly, we present the formal problem definition.

\begin{problem} 
\label{problem:formal}
    Let a  \ac{lpomdp} $\calM$ and a finite horizon $T$ be given. Let $\Pi$ be a policy space.
    Compute a policy $\pi \in \Pi$ that minimizes the conditional entropy of $Z_T$  given the partial information $Y_{0:T}= (O_{0:T}, A_{0:T})$ induced by $\pi$ and maximizing the probability of satisfying the task. That is,
\begin{equation*}
\optmin_{\pi \in \Pi} \quad H(Z_T  |Y_{0:T};{\calM_\pi}) -\alpha P^\pi(W_T=1),
\end{equation*}
where $\alpha>0$ is a weight parameter, for balancing the two objectives, and
the conditional entropy term 
\begin{equation}
\begin{aligned}
& H(Z_T  |Y_{0:T}; {M_\pi}) \coloneqq \\  
& - \sum_{\substack{z_T \in \mathcal{Z}, y_{0:T} \in \calO^T\times \calA^T}}  & P^\pi (z_T, y_{0:T} ) \cdot \log P^\pi (z_T \mid y_{0:T}) 
\end{aligned}
\end{equation}
and $P^\pi (z_T, y_{0:T} )$ is the joint probability of reaching a secret state and observing $y_{0:T}$ under the policy $\pi$ and $P^\pi (z_T \mid y_{0:T})$ is the conditional probability of reaching a secret state given observation $y_{0:T}$. 
\end{problem}

\begin{remark}When  $F_{sec} = F $ in the automaton $\dfa$, the planning objective is to maximize the probability of task satisfaction while minimizing the uncertainty in determining whether it is satisfied from partial observations. In this case, the active perception objective enables  the agent to actively gather information to improve its understanding of the task progress. In scenario when the optimal policy can complete the task with 50\% probability, the active perception agent will have more information to infer if a task is satisfied or not in each run from the partial observations.
We consider a more general case when the two sets can be different.
\end{remark}



\section{COMPUTING THE JOINT CONTROL AND ACTIVE PERCEPTION POLICY}
\label{sec:compute}
In this section, we develop a policy gradient method to solve Problem~\ref{problem:formal}. Consider a class of parametrized (stochastic) policies $\{\pi_\theta\mid \theta\in \Theta\}$. We denote by $\calM_\theta$ the stochastic process $\{V_t,A_t,O_t, t\in \nat\}$ induced by a policy $\pi_\theta$, and $P_\theta(\cdot)$ the corresponding probability measure. We use $Y$ as abbreviation of $Y_{0:T}$ and $y$ as abbreviation of $y_{0:T}$ in the following. 

To obtain the locally optimal policy parameter $\theta$, we initialize a policy parameter $\theta_0$ and carry out the gradient descent. At each iteration $\tau\ge 1$, 
\begin{equation}
\label{eq:gradient_decsent_algorithm}
\theta_{\tau + 1} = \theta_{\tau} - \eta  [\nabla_\theta H (Z_T  |Y;\theta_{\tau}) - \alpha \nabla_\theta P_{\theta_\tau}(W_T = 1)],
\end{equation}
where $\eta$ is the step size and $H(Z_T|Y; \theta) \coloneqq H(Z_T |Y; \calM_{\pi_\theta})$ and $P_\theta(W_T = 1) = P^{\pi_\theta}(W_T = 1)$. The following theorem gives the gradient of the conditional entropy $H(Z_T |Y; \theta)$, whose sub-step will also provide a method to compute the gradient of $P_\theta(W_T = 1)$ with respect to $\theta$.

\begin{lemma}
\label{lma:entropy-grad}
The gradient of the conditional entropy w.r.t. the policy parameter $\theta$ is 
\begin{equation}
\label{eq:entropy-grad-expectation}
\begin{aligned}
&\nabla_\theta H (Z_T |Y;\theta) =  \expect_{y\sim \calM_\theta} \left[ H(Z_T|Y= y)\nabla_\theta \log P_\theta(y )\right].
\end{aligned}
\end{equation}
\end{lemma}

\begin{proof}
The conditional entropy $H (Z_T |Y;\theta)$ can be written as
\begin{equation}
H(Z_T |Y;\theta) = \sum_y H(Z_T |Y = y) P_\theta(y).
\end{equation}
Note that the entropy term $H(Z_T |Y = y)$ does not depend on policy parameter $\theta$ since the observation trajectory $y$ is given and includes an action sequence. Then the gradient of $H(Z_T |Y;\theta)$ is given by
\begin{equation}
\nabla_\theta H(Z_T |Y;\theta) =  \sum_y H(Z_T |Y = y) \nabla_\theta P_\theta(y).
\end{equation}
By the logarithm trick $\nabla_\theta P_\theta(y) = P_\theta(y) \nabla_\theta \log P_\theta(y)$, we obtain
\begin{equation}
\begin{aligned}
\nabla_\theta & H(Z_T |Y;\theta) = \sum_y P_\theta(y) H(Z_T |Y = y) \nabla_\theta \log P_\theta(y) \\
& = \expect_{y\sim \calM_\theta} \left[ H(Z_T|Y= y)\nabla_\theta \log P_\theta(y )\right],
\end{aligned}
\end{equation}
which completes the proof.
\end{proof}


To calculate the gradient $H(Z_T|Y;\theta)$, we introduce the \emph{observable operator} augmented with control actions. The observable operator \cite{jaeger2000observableoperator} has been proposed to represent a discrete \ac{hmm} and used to calculate the probability of an observation sequence in an \ac{hmm} using matrix multiplications.

Let the random variable of state, observation, and action, at time point $t$ be denoted as $V_t, O_t, A_t$, respectively. 
We introduce an indexing of the product state set $\cal V$. Let $N =\lvert{\mathcal{V}\rvert}$ be the total number of states. The transition function $\Delta$ in the product \ac{pomdp} with domain $\mathcal{V}\times \calA$ and codomain $\dist{\mathcal{V}}$ is directly extended to the domain $\{1,\ldots, N\}\times A$ and co-domain $\dist{\{1,\ldots, N\}}$  using the state indices.

Consider an action $a\in \calA$, the $i,j$-th entry of the transposed transition matrix $\matT^a \in \reals^{N\times N }$   is defined to be:
\[
\matT_{i,j}^a = \Delta(i|j,a), \forall i , j \in \{1,\ldots, N\} 
\]
which is the probability of reaching state $i$ from state $j$ with action $a$.
Similarly, we index the observation set $\cal O$ as $\{1,\ldots, M\}$ and refer to each $o \in \calO$ by its index. 
For each $a\in \calA$, Let $\matO^a \in \reals^{M \times N}$ be the observation probability matrix with $\matO_{o,j}^{a} =E( o |  j, a)$.

\begin{definition} 
Given the product \ac{pomdp} $\calM$, for any pair of observation and  action $(o,a)$,
 the observable operator given  actions $\matA_{o|a}$ is a matrix of size $N \times N$ with its $ij$-th entry defined as $$
 \matA_{o |a }[i,j] =  \matT_{i , j}^a \matO_{o,j}^{a} \ ,$$
 which is the probability of transitioning from state $j$ to state $i$ and at the state $j$,  an observation $o$ is emitted given   perception action $a$.
  In matrix form, 
\[
\matA_{o |a } = \matT^a \text{diag}(\matO_{o, 1}^{a }, \dots, \matO_{o , N}^{a }).
\]
\end{definition}

The following propositions~\ref{prop:observation-action-probability}, \ref{prop:observation_prob}, \ref{prop:two_grads_equal} are   extensions of  our prior work \cite{shi2024active,udupa2025}. The difference is that in \cite{shi2024active,udupa2025}, action does not affect the state transition. In this current problem, action $a$ affects both the transition matrix $\matT^a$ and the observation matrix $\matO^a$. 
\begin{proposition}\cite{shi2024active,udupa2025}
\label{prop:observation-action-probability} The probability of an observation sequence $o_{0:t}$ given a sequence of   actions $a_{0:t}$, can be written as matrix operations,
\begin{equation}
\label{eq:matrix_operation}
P(o_{0:t} | a_{0:t}) = \mathbf{1}_N^\top \matA_{o_t|a_t} \dots \matA_{o_0|a_0} \mu_0.
\end{equation}
In addition, 
for a fixed state $v_{t + 1} \in \mathcal{V}$ at time point $t + 1$, we have 
\begin{equation}
\label{eq:matrix_operation_sT}
P(v_{t+1}, o_{0:t} | a_{0:t}) = \mathbf{1}_{v_{t+1}}^\top \matA_{o_t|a_t} \dots \matA_{o_0|a_0} \mu_0.
\end{equation}
where $\mathbf{1}_{v_{t+1}}$ is a one-hot vector which assigns 1 to the $v_{t+1}$-th entry.
\end{proposition}

To compute the gradient $\nabla_\theta H(Z_T|Y;\theta)$, we will need the value of $\nabla_\theta \log P_\theta(y)$. We start by calculating the probability of an observation sequence $y=(o_{0:t}, a_{0:t}) $ in $\calM_\theta$. 
\begin{proposition} 
\label{prop:observation_prob}
\cite{shi2024active}
The probability of a sequence  $y=(o_{0:t}, a_{0:t}) $ of observations and actions in $\calM_\theta$ can be computed as follows:
\begin{equation}
\begin{aligned}
\label{eq:joint_final}
P_\theta(o_{0:t}, a_{0:t}) & =  \frac{P(o_{0:t} | a_{0:t}) }{P(o_{0} | a_{0}) }  \prod_{i=0}^{t} \pi_\theta(a_{i}| o_{0:i-1}). 
\end{aligned}
\end{equation}
where $o_{0:-1} \coloneqq \varnothing$ is the initial empty  observation.
\end{proposition}

\begin{proposition}\cite{shi2024active}
\label{prop:two_grads_equal}
Given $y= (o_{0:T}, a_{0:T})$, 
the gradient of $\log P_\theta(y|s_0)$ can be computed as:
\begin{equation}
\label{eq:two_grads_equals}
\begin{aligned}
\nabla_\theta \log  P_\theta(y )=  \sum_{t = 0}^{T} 
\nabla_\theta \log  \pi_\theta(a_{t}| o_{0:t-1}).
\end{aligned}
\end{equation}
\end{proposition}

We also need the value of entropy $H(Z_T|Y = y)$ to compute the gradient $\nabla_\theta H(Z_T|Y;\theta)$. It can be written as
\begin{equation}
H(Z_T|Y = y) = \sum_{Z_T \in \mathcal{Z}} P(z_T|y) \log P(z_T|y).
\end{equation}


\begin{proposition}
Given $y= (o_{0:T}, a_{0:T})$, 
the probability $P(Z_{T} = 1| y)$ does not depend on the policy parameter $\theta$ since the actions are given. It can be computed as 
 
\begin{equation}
\label{eq:P_zT1_y}
\begin{aligned}
 P(Z_{T} & = 1| y) = \sum_{v_{T} \in \calF_{sec}} \obs(o_T|v_T, a_T) \\ 
& \cdot  P(v_{T}, o_{0:T-1}| a_{0:T-1}) /  P(o_{0:T}| a_{0:T}),
\end{aligned}
\end{equation}
where  the probability $P_\theta(v_{T}, o_{0:T-1}| a_{0:T-1})$ can be calculated by equation~\eqref{eq:matrix_operation_sT} and the probability $P_\theta(o_{0:T-1}| a_{0:T-1})$ can be calculated by equation~\eqref{eq:matrix_operation}.
\end{proposition}
\begin{proof}
First, the probability $P(Z_{T} = 1 | y)$ can be written as
\begin{equation}
\label{eq:prop4_eq0}
P(Z_{T} = 1 | y) = \sum_{v_{T} \in \calF_{sec}} P(v_{T} | y).
\end{equation}
The probability
\begin{equation}
\label{eq:prop4_eq1}
\begin{aligned}
P(v_{T}| y) = \frac{P(v_T, o_T| o_{0:T-1}, a_{0:T})}{P(o_T| o_{0:T-1}, a_{0:T})},
\end{aligned}
\end{equation} 
whose numerator can be rewritten as follows,
\begin{equation}
\label{eq:numerator-0}
\begin{aligned}
P(v_T, o_T | o_{0:T-1}, a_{0:T})& = P(o_T|v_T, o_{0:T-1}, a_{0:T}) \\
& \cdot P(v_T| o_{0:T-1}, a_{0:T}). 
\end{aligned}
\end{equation}
Since the observation $o_T$ only depends on $v_T, a_T$, we have
\begin{equation}
\label{eq:numerator-1}
P(o_T|v_T, o_{0:T-1}, a_{0:T}) = \obs(o_T|v_T,a_T).
\end{equation}
And since $v_T$ and $a_T$ are conditionally independent on observations $ o_{0:T-1}$ and the past action sequence $a_{0:T-1}$, we obtain
\begin{equation}
P(v_T| o_{0:T-1}, a_{0:T}) = P(v_T| o_{0:T-1}, a_{0:T-1}). 
\end{equation}
Furthermore, 
\begin{equation}
\label{eq:numerator-2}
P(v_{T}| o_{0:T-1}, a_{0:T-1}) = \frac{P(v_{T}, o_{0:T-1}| a_{0:T-1})}{P(o_{0:T-1}| a_{0:T-1})},
\end{equation} Substituting \eqref{eq:numerator-1} and \eqref{eq:numerator-2} into \eqref{eq:numerator-0}, the numerator becomes
\begin{equation}
\label{eq:numerator}
\begin{aligned}
P(v_T,& o_T | o_{0:T-1}, a_{0:T}) = \obs(o_T|v_T, a_T) \\ 
& \cdot  P(v_{T}, o_{0:T-1}| a_{0:T-1}) /  P(o_{0:T-1}| a_{0:T-1}).
\end{aligned}
\end{equation}
And the denominator
\begin{equation}
P(o_T| o_{0:T-1}, a_{0:T}) = \frac{P(o_{0:T}| a_{0:T})}{P(o_{0:T-1}| a_{0:T})}.
\end{equation}
The observations $o_{0:T-1}$ does not depend on the future action $a_T$, thus we have
\begin{equation}
\label{eq:denominator}
P(o_T| o_{0:T-1}, a_{0:T}) = \frac{P(o_{0:T}| a_{0:T})}{P(o_{0:T-1}| a_{0:T-1})}.
\end{equation}
Substitute the equation~\eqref{eq:numerator} and \eqref{eq:denominator} into equation~\eqref{eq:prop4_eq1}, we can cancel the term $P(o_{0:T-1}| a_{0:T-1})$ and obtain 
\begin{equation}
\begin{aligned}
 P(Z_{T} & = 1| y) = \sum_{v_{T} \in \calF_{sec}} \obs(o_T|v_T, a_T) \\ 
& \cdot  P(v_{T}, o_{0:T-1}| a_{0:T-1}) /  P(o_{0:T}| a_{0:T}).
\end{aligned}
\end{equation}
\end{proof}

The following proposition gives the gradient of the probability $P_\theta(W_T = 1)$.

\begin{proposition}
the gradient of the probability $P_\theta(W_T = 1)$ w.r.t. policy parameter $\theta$ is
\begin{equation}
\label{eq:success-gradient}
\nabla_\theta P_\theta(W_T = 1) = \sum_y P(W_T = 1|y) P_\theta(y) \nabla_\theta \log P_\theta(y). 
\end{equation}
\end{proposition}

\begin{proof}
By the law of total probability
\begin{equation}
P_\theta(W_T = 1) = \sum_y  P(W_T = 1|y) P_\theta(y).
\end{equation}
Note that $P(W_T = 1|y)$ does not depend on $\theta$ since all actions are given in observations $y$. Then the gradient of $P(W_T = 1|y)$ is 
\begin{equation*}
\nabla_\theta P_\theta(W_T = 1) = \sum_y  P(W_T = 1|y) \nabla_\theta P_\theta(y).
\end{equation*}
By using the logarithm trick, we can obtain $\eqref{eq:success-gradient}$. 
\end{proof}

According to all the propositions and theorems above, we can calculate the gradient $\nabla_\theta H(Z_T|Y;\theta)$ and $\nabla_\theta P_\theta(W_T = 1)$. With the gradient computation, a gradient descent algorithm can be employed to solve a locally optimal solution for Problem~\ref{problem:formal}.


It is noted that, though $\mathcal{O}^T$ and $\mathcal{A}^T$ are finite sets of observations and actions, They are combinatorial and may be too large to enumerate. To mitigate this issue, we can employ sample approximations to estimate $\nabla_\theta H(Z_T|Y;\theta)$ and $\nabla_\theta P_\theta(Z_T = 1)$: 
Given $M$ sequences of observations $\{y_1, \dots, y_M\}$, we can approximate $ H(Z_T|Y;\theta)$ by
\begin{equation}
\hat{\nabla}_\theta H (Z_T |Y;\theta) =  \frac{1}{M} \sum_{k = 1}^M H(Z_T|Y= y_k)\nabla_\theta \log P_\theta(y_k)
\end{equation}
for $k = 1, \dots, M$. And
\begin{equation}
\hat{\nabla}_\theta P_\theta(Z_T = 1) = \frac{1}{M} \sum_y P(Z_T = 1|y_k) \nabla_\theta \log P_\theta(y_k). 
\end{equation}

\begin{figure}[t!]
\centering
\begin{tikzpicture}[->,>=stealth',shorten >=1pt,auto,node distance=2.5cm, scale=0.8,transform shape]
    \node[state] (0) {\Large $0$};
    \node[state] (1) [right=1.5cm of 0] {\Large $1$};
    \node[state] (2) [below=2.5cm of 1] {\Large $2$};
    \node[state] (3) [right=2.5cm of 1] {\Large $3$};
    \node[state] (4) [right=2.5cm of 2] {\Large $4$};
    \node[state] (5) [right=1.5cm of 3] {\Large $5$};
    
    \path 	
        (0) edge node {$a, b :0.5$} (1)
        (0) edge node [left] {$a, b :0.5$} (2)
        (0) edge [loop left] node{$c$} (0)
        (1) edge [bend left] node {$a, b: 0.5$} (3)
        (1) edge node [right, pos = 0.75] {$a, b: 0.5$} (4)
        (1) edge [loop above] node{$c$} (1)
        (2) edge node [pos=0.2] {$a, b: 0.5$} (3)
        (2) edge [bend right] node [below] {$a, b: 0.5$} (4)
        (2) edge [loop below] node{$c$} (2)
        (3) edge [bend left] node {$a$} (5)
        (3) edge node [above] {$b:0.9$} (1)
        (3) edge [bend left] node [right, pos=0.15] {$a:0.1$} (2)
        (3) edge [loop above] node{$c$} (3)
        (4) edge node {$a$} (5)
        (4) edge [bend left] node [left, pos=0.85] {$b:0.1$} (1)
        (4) edge node {$a:0.9$} (2)
        (4) edge [loop below] node{$c$} (4)
        (5) edge node [align=center] {$a: 0.9$ \\ $b: 0.1$} (3)
        (5) edge [bend left] node [align=center] {$a: 0.1$ \\ $b: 0.9$} (4)
        (5) edge [loop right] node {$c$} (5);
\end{tikzpicture}
\caption{The graph represents the transition of the \ac{pomdp}. The arrows labeled with \(a\), \(b\), or \(c\) represent deterministic actions, while the arrows labeled with both actions and probabilities indicate stochastic actions that lead to a particular state with a given probability.}
\label{fig:graph_example}
\end{figure}
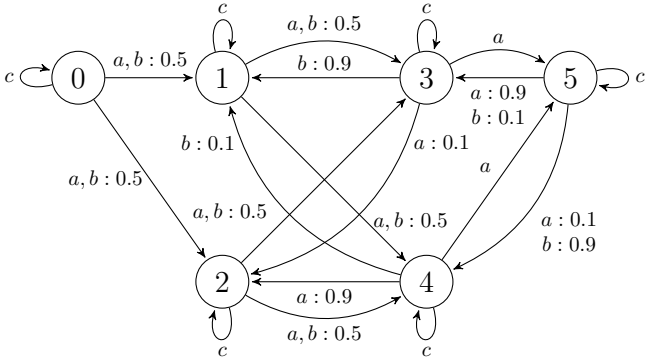

\section{EXPERIMENTAL EVALUATION}
\label{sec:simulation}

We illustrate the effectiveness of the proposed joint control and active perception algorithm using experiments. A UAV (perception agent) and a moving ground robot operate within the same environment. The dynamic models of both the UAV and the ground robot are represented using the same graph-like \ac{mdp}, shown in Fig.~\ref{fig:graph_example}. It consists of six nodes $\{0, 1, 2, 3, 4, 5\}$ and three actions $\{a, b, c\}$. That, the UAV can be at one of the six nodes and take one of the three actions. So is the round robot. The two agents move concurrently.

There are two types of ground robots—a nominal robot and an adversarial robot—each with distinct goal states. The goal states for the nominal robot are $\{3, 4\}$, while the adversary aims for node $\{4\}$. Once the ground robot reaches its goal states, it stay there and continuously collect rewards. Therefore, the optimal policy for the nominal robot is more random: It can go through either node $1$ or $2$ to reach node $3$ or $4$ and stay there. However, for the adversary robot, it will directly move to node $4$ by passing either node $1$ or $2$. We solve two optimal policies, for the nominal robot and the adversarial robot's having different reward functions.

The state of the \ac{pomdp} is represented as $(s_g, s_u, \text{ty})$, where $s_g \in \{0, 1, 2, 3, 4, 5\}$ denotes the ground robot’s state, $s_u \in \{0, 1, 2, 3, 4, 5\}$ represents the UAV’s state, and $\text{ty}$ indicates the robot type. The state space of the \ac{pomdp} is denoted as $\mathcal{S}$, and the action space, $\{a, b, c\}$, corresponds to the UAV’s actions since the ground robot's movements are not controllable and pre-defined with two defaulty policies. The \ac{pomdp} transition dynamics are derived from the independent transitions of both the UAV and the ground robots given the policy of the nominal robot and the policy of the adversarial robot.  

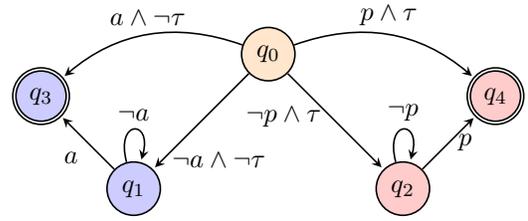
\begin{figure}[t!]
\centering
\begin{tikzpicture}[->, >=stealth, shorten >=1pt, auto, node distance=1cm, semithick]
\node[draw, circle, fill=orange!20] (0) {$q_0$};
\node[draw, circle, below left= 1.8cm of 0,fill=  blue!20] (1) {$q_1$};
\node[draw, circle, below right= 1.8cm of 0,fill=red!20] (2) {$q_2$};
\node[draw, circle, above left= of 1,fill=blue!20,accepting] (3) {$q_3$};
\node[draw, circle, above right= of 2,fill=red!20, accepting] (4) {$q_4$};
\path[->] (0) edge node [right, pos = 0.9] {$\neg a \land \neg \tau$} (1)
(0) edge [bend right] node [above] {$a \land \neg \tau$} (3)
(0) edge [bend left] node [above]  {$p \land \tau$} (4)
(1) edge [loop above] node {$\neg a$} (1)
(1) edge node {$a$} (3)
(0) edge node [left, pos = 0.45] {$\neg p \land \tau$} (2)
(2) edge [loop above] node {$\neg p$} (2)
(2) edge node [right] {$p$} (4);
\end{tikzpicture}
\caption{The \ac{dfa} for the task. Self-loops with label $\varnothing$ are omitted. Double circle nodes represent accepting states.}
\label{fig:DFA}
\end{figure}

The UAV is equipped with a camera that can detect the ground robot’s state when they are on the same node, with a detection probability $0.9$. The UAV receives a null observation ($n$) with probability $0.1$ when the detection fails. The observation space of the \ac{pomdp} includes all possible ground robot states along with the null observation. There is no stationary sensor in the environment.
Initially, the ground robot starts at node $0$, while the UAV starts at node $5$. 

The UAV is given the following task specification:
\begin{itemize}
\item[1.] If the ground robot is the nominal robot, then the UAV perform surveillance task $\varphi_0 \coloneqq \Eventually a$, where $a$ is evaluated true if the UAV reaches the state $2$.
\item[2.] If the ground robot is the adversary, then the UAV performs interdiction task $\varphi_1= \Eventually p$ , where $p$ is evaluated true if the UAV is on the same node as the ground robot (The UAV eventually interdicts the adversarial robot). 
\end{itemize}
The task can be represented by the following \ac{ltlf} formula
$$(\neg \tau \rightarrow \varphi_0   \land   \tau \rightarrow \varphi_1),$$
where $\tau$ is an atomic proposition that indicates the type of ground robots. $\tau$ is evaluated true if the ground robot is an adversary. Fig.~\ref{fig:DFA} shows the corresponding \ac{dfa} of the \ac{ltlf} formula. The formula is evaluated over a finite mission time $T = 5$, based on the semantics of \ac{ltlf}. This \ac{dfa} is not the minimal one, because we would like to introduce the following inference task using the same \ac{dfa}. 

The secret to be inferred is the type of the agent. In the \ac{dfa} in Fig.~\ref{fig:DFA}, these blue states are the automata states that can be visited only if the ground robot is the nominal one; and the red states are the automata states that can be visited only if the ground robot is adversarial. Thus, $F_{sec} = \{q_2,q_4\}$ and $F = \{q_3, q_4\}$. If the last automaton state is in $F_{sec}$, then the trajectory is generated by interacting with an adversarial robot. If the last automaton state is in $F$, then the trajectory satisfies the task, regardless if the robot is nominal or adversarial.


We employ a \emph{finite-state, observation-based} perception policy $\pi$ with deterministic transitions, defined as a tuple $\pi \coloneqq \langle \mathbb{X}, \mathbb{O}, \mathcal{A}, \mathbb{T}, \psi, x_0 \rangle $ where \begin{inparaenum}
\item $ \mathbb{X}$ is a set of memory states.
\item $\mathbb{O}, \mathcal{A}$ are a set of inputs and a set of outputs, respectively.
\item $\mathbb{T}: \mathbb{X} \times \mathbb{O} \rightarrow \mathbb{X}$ is a deterministic transition function that maps a state and an input, $(x,o)$, to a next state, $\mathbb{T}(x,o)$. 
\item $\psi:\mathbb{X} \rightarrow \dist{\mathcal{A}} $ is a probabilistic output function.
\item $x_0$ is the initial state.
\end{inparaenum}
Specifically, given an integer $K \ge 0$, the state set $\mathbb{X}=\{\mathbb{O}^{\le K}\}$ are the set of observations with length $\le K$.  For each $x\in \mathbb{X}$, $\delta(x, o) = x'$ is defined such that $x' =  \mathsf{suffix}^{= K}(x\cdot o)$ \footnote{$\mathsf{suffix}^{= K}(w)$ is the last $K$ symbols of string $w$ if $|w|\ge K$ or $w$ itself otherwise.} is the last (up to) $ K$  observations after appending the observation $o$ to  $x$. The probabilistic output function is parameterized as, 
$\psi_\theta(a|x) =  \frac{\exp(\theta_{x,a})}{\sum_{a' \in \mathcal{A}}\exp(\theta_{x, a'})},
$ where  $\theta \in \reals^{
|\mathbb{X} \times \mathcal{A}|}$ is the policy parameter vector. Given observations $o_{0:t}$, the policy $\pi_\theta(a| o_{0:t}) =  \psi_\theta( a| \delta(x_0, o_{0:t}))$.
The softmax policy is differentiable and complete. 
In the experiments, we set the length of memory $K=2$.

\begin{figure}[t]
\centering
\includegraphics[width=\linewidth]{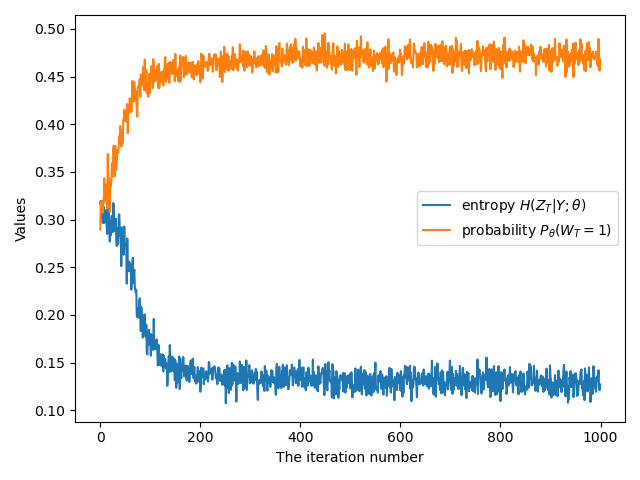}
\caption{The convergence results of the policy gradient method when goal states of different types are overlapping.}
\label{fig:opt_results}
\end{figure}

The prior distribution $\mu_0$ is set to be $0.5$ for each type of agent. Figure~\ref{fig:opt_results} illustrates the convergence of the policy gradient \footnote{We sample $M = 1000$ trajectories and set the horizon $T = 5$ for each iteration. The fixed step size of the gradient descent algorithm is set to be $0.5$. We run $N = 1000$ iterations on the 12th Gen Intel(R) Core(TM) i7-12700, the average time consumed for one iteration is $4.7$ seconds.}. We denote the policy parameter of optimal policy as $\theta^\star$.
When the algorithm converges, the conditional entropy $ H(Z_T|Y;\theta^\star)$ 
approaches approximately $0.127$. This indicates that the observations provide substantial information about the robot's type on average. And the probability of completing the task $P_{\theta^\star}(W_T = 1)$ reaches approximately $0.468$. 
Because the UAV struggles to identify the type of ground robot when it could be either of the two types, task completion becomes more challenging. However, if the ground robot is set to be only one of the types, the probability of completing the task will increase under the same policy. We denote the prior distribution that the probability of being the nominal robot is $1$ as $\mu_0^0$ and the prior distribution that the probability of being the adversary is $1$ as $\mu_0^1$. The probability of completing the task under the prior  distribution $\mu_0^0$ is $P_{\theta^\star}(W_T = 1; \mu_0^0) = 0.667$. Similarly, if the ground robot is set to be only an adversary, $P_{\theta^\star}(W_T = 1; \mu_0^1)= 0.852$. This improvement occurs because the UAV can determine more easily  the robot’s type and select the appropriate actions when there is no ambiguity in classification. Note that the above performance is evaluated using a policy optimized for a different prior (originally $\mu_0$ is set to be $0.5$ for each type). If the UAV initially knows with certainty whether the agent is nominal or adversarial, it can compute policies using the corresponding priors $\mu_0^0$ and $\mu_0^1$, leading to higher task completion probabilities. However, these policies may perform poorly if the prior is $\mu_0$.

\begin{figure}[t]
\centering
\includegraphics[width=\linewidth]{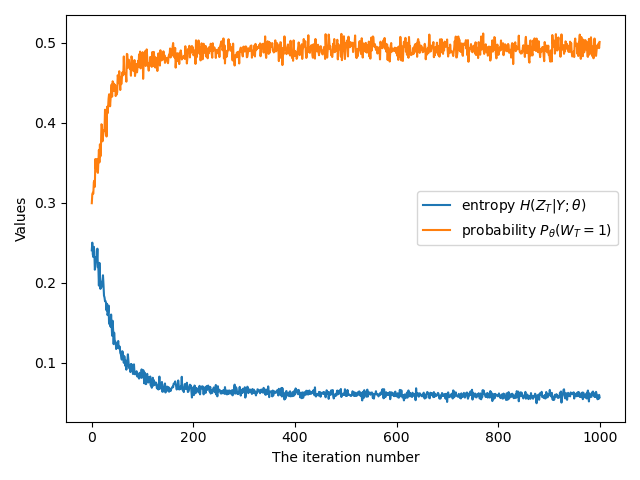}
\caption{The convergence results of the policy gradient method when goal states of different types are non-overlapping.}
\label{fig:opt_results_2}
\end{figure}
In the previous simulation, we set the goal states of two different types of ground robots to overlap ($\{3, 4\}$ and $\{4\}$). In contrast, when the goal states are non-overlapping: the goal state for the nominal robot is  $\{3\}$, while the adversary aims for node $\{4\}$.
We obtain the results in Fig.~\ref{fig:opt_results_2}. As the algorithm converges, the conditional entropy $ H(Z_T|Y;\theta) $ approaches approximately $0.056$, indicating that the observations provide more information about the robot's type than the previous results. Additionally, the probability of task completion, $P_{\theta^\star}(W_T = 1)$, reaches approximately $0.501$.  
The decrease in entropy suggests that identifying the type of ground robot becomes easier for the UAV due to the different, non-overlapping goal sets between the two types of robots.
When the identification becomes easier, the probability of task completion increases. This is because the UAV can more effectively determine the type of ground robot and then select appropriate actions to complete the corresponding task. It is noted that the UAV needs not to infer the agent's type and then plan with the identified type. It only needs to maximizing the probability of completing the formula---thanks to the expressiveness of temporal logic.

\section{CONCLUSION AND FUTURE WORK}
In this paper, we designed an algorithm to obtain an optimal joint control and active perception policy that influences both the dynamics and the emission function in a \ac{pomdp} model. By using conditional entropy to quantify uncertainty, we formulated an optimization problem that minimizes the uncertainty about some critical temporal events while maximizing task completion probability in   stochastic, partially observable environments. To solve this problem, we developed a gradient descent algorithm and derived the gradient of the objective function using observable operators, resulting in a gradient form similar to the policy gradient theorem \cite{Sutton1999policy}.  

Future research could explore joint control and active perception under different assumptions for the perception agent, such as scenarios where the agent has imprecise knowledge of the model dynamics. Another promising direction is extending this work to systems with continuous observations, such as camera images and  sensor data, to broaden the applicability of our approach to real-world tasks.


\bibliographystyle{IEEEtranS}
\bibliography{ref}

\end{document}